\documentclass[a4paper]{article}
\usepackage[utf8]{inputenc}
\usepackage{geometry}
\usepackage{xcolor}
\usepackage{pst-node}
\usepackage{graphicx}
\usepackage{amsmath}
\usepackage{mathtools}
\usepackage{kbordermatrix}
\usepackage{hyperref}
\usepackage[shortlabels]{enumitem}
\usepackage{dsfont}
\usepackage{amsthm}
\usepackage{amssymb}
\usepackage{amstext}
\usepackage{amsfonts}
\usepackage{nicefrac}
\usepackage{extarrows} 
\title{Two-electron wavefunctions are 
matrix product states with 
bond dimension Three}

\author{Gero Friesecke and Benedikt R. Graswald \\[1mm] 
\small Department of Mathematics, Technical University of Munich, Germany \\
\small {\tt \href{mailto:gf@ma.tum.de}{gf@ma.tum.de}, \href{mailto:graswabe@ma.tum.de}{graswabe@ma.tum.de} }}

\renewcommand{\kbldelim}{(}
\renewcommand{\kbrdelim}{)}

\renewcommand{\phi}{\varphi}
\renewcommand{\epsilon}{\varepsilon}

\DeclareMathOperator{\Id}{Id}

\DeclareMathOperator{\rk}{rank}

\newcommand{\gegen}[2]{\xrightarrow{#1 \to #2}}
\newcommand{\MPS}{\mathrm{MPS}}

\newcommand{\limit}[2]{\lim\limits_{#1 \to #2}}

\newcommand{\R}{\mathbb{R}}
\newcommand{\C}{\mathbb{C}}

\newcommand{\N}{\mathcal{N}}

\newcommand{\V}{\mathcal{V}}
\renewcommand{\H}{\mathcal{H}}
\newcommand{\F}{\mathcal{F}}
\newcommand{\norm}[1]{\left\lVert#1\right\rVert}
\newcommand{\bigzero}{\mbox{\normalfont\Large\bfseries 0}}
\newcommand{\rvline}{\hspace*{-\arraycolsep}\vline\hspace*{-\arraycolsep}}

\newtheorem{theorem}{Theorem}

\newtheorem*{theorem*}{Theorem}
\newtheorem{lemma}{Lemma}

\newtheorem{proposition}{Proposition}
\newtheorem{example}{Example}

\begin{document}
\large
\maketitle
\begin{abstract}
We prove 
the statement in the title,
for a suitable (wavefunction-dependent) choice of the underlying orbitals,  
and show that Three is optimal.
Thus for two-electron systems, the QC-DMRG method with bond dimension Three combined with fermionic mode optimization exactly recovers the FCI energy.
\end{abstract}

\section{Introduction}

The $N$-electron Schr\"odinger equation is a partial differential equation in $\R^{3N}$ and its direct numerical solution is prohibited for large $N$ by the curse of dimension. As a consequence, a large variety of approximate methods have been developed since the early days of quantum mechanics, starting with the work of Thomas, Fermi, Dirac, Hartree, and Fock. In the past decade, the Quantum Chemistry Density Matrix Renormalization Group (QC-DMRG) method \cite{white_martin, mitrushenkov2001quantum, chan2002highly, legeza2003controlling} has become the state-of-the-art choice for systems with up to a few dozen electrons; see \cite{szalay2015tensor} for a recent review.  

In QC-DMRG, one chooses a suitable finite single-particle basis, makes a matrix product state (MPS) alias tensor train ansatz for the coefficient tensor of the many-particle wavefunction in Fock space, and optimizes the Rayleigh quotient over the matrices (see Sections \ref{sec:mps} and \ref{sec:two_particle_finite} for a detailed description). The key parameter in the method is the maximal allowed size of the matrices, called bond dimension. For bond dimension 1 the MPS ansatz reduces to a single Slater determinant built from the basis functions. For large bond dimension the ansatz recovers all wavefunctions in the Fock space, but large means impractically large (more precisely: $2^{L/2}$, where $L$ is the number of single-particle basis functions \cite{schollwock2011density}).

It has long been known that the accuracy strongly depends on the choice of basis, and can typically be improved by re-ordering the basis (see \cite{BLMR, dupuy2021inversion, szalay2015tensor}; also, see \cite{graswald2021electronic} for extreme examples where ordering does not yield an improvement). 

This paper is motivated by an empirical phenomenon observed by Krumnow, Veis, Legeza, and Eisert \cite{krumnow-legeza2016, krumnow-legeza2021}: going beyond ordering and {\it optimizing over fermionic mode transformations} (i.e., general unitary transformations of the single-particle basis) can reduce the approximation error a great deal further in systems of interest. QC-DMRG together with optimization over the single-particle basis as introduced in \cite{krumnow-legeza2016, krumnow-legeza2021} can be viewed as a generalization of the classical Hartree-Fock method, to which it reduces for bond dimension 1 (see Section \ref{sec:two_particle_finite}). 

In the absence of previous mathematical results on the influence of mode transformations on the approximation error, we investigate here the simplest case $N=2$. We find a dramatic effect, namely a {\it reduction of the bond dimension needed for exactness of the method from $2\! +\!\tfrac{L}{2}$ to $3$},  where $L$ is the number of single-particle basis functions (Theorem \ref{thm:exact} in Section \ref{sec:two_particle_finite}). This is proved by showing that general two-particle wavefunctions can be represented exactly with bond dimension $3$ after a (wavefunction-dependent) optimal mode transformation, with $3$ being optimal. See Theorems \ref{thm:upper_bound} and \ref{thm:optimality} in Section  
\ref{sec:two_particle_finite}. 

Previous exact representations of quantum states in the form of low-bond-dimension MPS were, to our knowledge, limited to very special states, the prototype example being the AKLT state from spin physics \cite{aklt} which arises as the ground state of a particular translation invariant Hamiltonian. On the other hand, the present result -- unlike that in \cite{aklt} -- is limited to $N=2$ (see the Conclusions for further discussion of this point).

Finally, we remark that the exact bond-dimension-three representation of two-fermion wavefunctions carries over to the infinite-dimensional single-particle Hilbert space $L^2(\R^3)\otimes\C^2$ of full two-electron quantum mechanics, as shown in the last part of this paper.

\section{Fock space and occupation representation} 

{\it Fermionic Fock space.} We first consider a finite dimensional single-particle Hilbert space $\H_L$, whose dimension we denote by $L$. The associated state space for a system of $N$ fermions is the $N$-fold antisymmetric product  $\mathcal{V}_{N,L}:= \bigwedge_{i=1}^N \H_L$, and the resulting Fock space is defined as the direct sum of the $N$-particle spaces, 
\begin{equation} \label{eq:Fock}
   \F_L := \bigoplus_{N=0}^L {\mathcal V}_{N,L},
\end{equation}
where $V_{0,L}\widetilde{=}\C$ is spanned by the vacuum state $\Omega$. When the particles are electrons, $\H_L$ would correspond to a subspace of $L^2(\R^3)\otimes\C^2$ spanned by $L$ spin orbitals. If the orbitals are the occupied and lowest unoccupied eigenstates of the Hartree-Fock Hamiltonian associated with the electronic Schr\"odinger equation,  $\mathcal{V}_{N,L}$ is known in physics as the full configuration interaction (full CI) space (see e.g.~\cite{helgaker}).

Now given an orthonormal basis $\{\phi_1, \ldots, \phi_L\}$ of the single-particle Hilbert space $\H_L$, we can write any element $\Psi\in \F_L$ in the form
\begin{align} \label{eq:state_fci_representation}
\Psi = c_0 \Omega \; + \;  \sum_{i=1}^L c_i \varphi_i \; + \!\! \sum_{1\le i<j\le L} \!\!\! c_{ij}|\varphi_i\varphi_j\rangle \; + \!\! \sum_{1\le i<j<k\le L} \!\!\! c_{ijk} |\varphi_i\varphi_j\varphi_k\rangle  \; +  \; \ldots \; ,   
\end{align}
with $|\varphi_{i_1} \ldots \varphi_{i_N}\rangle$ denoting the antisymmetric tensor product alias Slater determinant
\begin{equation} \label{eq:sd}
   |\varphi_{i_1} \ldots \varphi_{i_N}\rangle = \varphi_{i_1}\wedge \ldots \wedge \varphi_{i_N} \in {\mathcal V}_{N,L}.
\end{equation}

{\it Occupation representation.}
Instead of the above 'first quantized' representation, in QC-DMRG one considers a 'second quantized' representation by occupation numbers of orbitals in Fock space. A Slater determinant $|\varphi_{i_1}...\varphi_{i_N}\rangle\in \V_{N,L}$ is represented by a binary string $(\mu_1, \ldots, \mu_L) \in \{0,1\}^L$, with $\mu_i$ indicating whether or not the orbital $\phi_i$ is present (occupied) or absent (unoccupied). An example with $N=4$ and $L=8$ is

\begin{minipage}{0.5 \textwidth}\vspace{-3.5mm}
\begin{align*}
    |\phi_2 \phi_3\phi_6\phi_8 \rangle \longleftrightarrow  (0,1,1,0,0,1,0,1), 
\end{align*}
\end{minipage}
\begin{minipage}{0.5 \textwidth}

\includegraphics[width = 0.8 \textwidth]{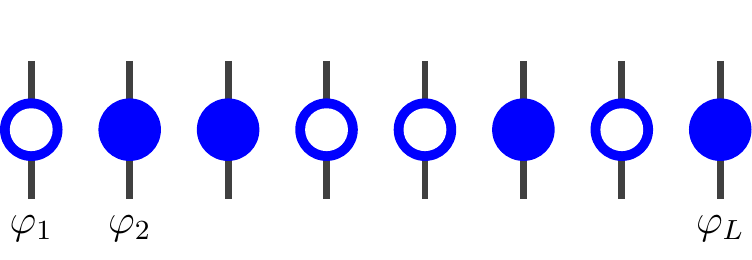}
\end{minipage}

\noindent
since $\varphi_1$ is unoccupied, $\varphi_2$ is occupied, $\varphi_3$ is occupied, and so on. The Slater determinant \eqref{eq:sd} indexed by its binary label is in the following denoted $\Phi_{\mu_1...\mu_L}$, that is to say
\begin{equation} \label{eq:sdbin}
    \Phi_{\mu_1...\mu_L} := |\varphi_{i_1}...\varphi_{i_N}\rangle \mbox{ if }\mu_{i}=1 \mbox{ exactly when }i\in\{i_1,...,i_N\},~ i_1<\ldots<i_N. 
\end{equation}
The coefficients in the expansion \eqref{eq:state_fci_representation} indexed by the corresponding binary label are called $C_{\mu_1...\mu_L}$, that is to say  
\begin{align}
     C_{\mu_1 \ldots \mu_L} 
     = 
     c_{i_1\ldots i_N} \mbox{ if } \mu_i=1 \text{ precisely when } i \in \{i_1,\ldots,i_N \},~i_1<\ldots<i_N,
\end{align}
yielding the occupation representation
\begin{equation} \label{eq:state_occrep}
  \Psi \; = \! \sum_{\mu_1,...,\mu_L=0}^1 \!\! C_{\mu_1...\mu_L} \Phi_{\mu_1...\mu_L}.
\end{equation}

\section{Matrix product states}\label{sec:mps}

A matrix product state (MPS) or tensor train (TT) with respect to the basis $\{\varphi_i\}_{i=1}^L$ 
with size parameters ('bond dimensions') $r_i$ ($i=1,...,L-1$) is a state of the form
\begin{equation} \label{eq:MPS}
   \Psi = \! \sum_{\mu_1,\ldots \mu_L =0}^1  \! A_1[\mu_1] A_2[\mu_2] ... A_L[\mu_L] ~\Phi_{\mu_1 ... \mu_L} \in \F_L
\end{equation}
where for every $(\mu_1,\ldots,\mu_L)$,  $A_i[\mu_i]$ is a  $r_{i-1} \times r_i$ matrix, with the convention $r_0=r_L=1$. 
Writing out the above matrix multiplications, 
\begin{align*}
    A_1[\mu_1] A_2[\mu_2]\ldots A_L[\mu_L]
    =
    \sum_{\alpha_1 =1}^{r_1} \sum_{\alpha_2 =1}^{r_2}\ldots \!\! \!\!\!\sum_{\alpha_{L-1} =L-1}^{r_{L-1}}\!\! \big(A_1[\mu_1]\big)_{\alpha_1} \big(A_2[\mu_2]\big)_{\alpha_1\alpha_2}\ldots \big(A_L[\mu_L]\big)_{\alpha_{L-1}}.
\end{align*}
Hence the $A_i$ can be viewed as tensors of order $3$ (depending on three indices $\alpha_{i-1}$, $\mu_i$, $\alpha_i$) in $\C^{r_{i-1}\times 2 \times r_i}$. The name 'bond dimensions' for the $r_i$ has nothing to do with chemical bonds, but is related to the standard graphical representation of MPS in Figure \ref{fig:tt_representation}, in which each contraction index $\alpha_i$ is represented by a horizontal 'bond'. The minimal bond dimensions with which a given state can be represented have a well known meaning as ranks of matricizations of the coefficient tensor $C$, as recalled in Lemma \ref{L:rank}. 
The set of tensor trains (TT) or matrix product states (MPS) with respect to the basis $\{\varphi_i\}_{i=1}^L$ with bond dimensions $r_i$ ($i=1,...,L-1$)  is denoted by 
\begin{equation}\label{eq:mps_representation_finite}
    \mathrm{MPS}\bigl(L,\{r_i\}_{i},\{\varphi_i\}_{i}\bigr) \; \; \subseteq \;\; \F_L.
\end{equation}
For bond dimension One, i.e.~$r_i=1$ for all $i$, the MPS set \eqref{eq:mps_representation_finite} reduces to the set of Slater determinants $\Phi_{\mu_1,...,\mu_L}$ built from the basis functions. 
Representing arbitrary states in $\F_L$ as MPS is possible, but requires bond dimensions $2^{L/2}$, i.e.~bond dimensions growing exponentially with $L$ \cite{schollwock2011density}. (Here we have assumed that $L$ is even.) A simple example where this exponential bound is saturated is the Slater determinant with orbitals $\psi_i : =  \big( \phi_i + \phi_{i+L/2} \big)/\sqrt{2}$ for $i=1,\ldots,L/2$, see e.g.~\cite{dupuy2021inversion, graswald2021electronic}. (In this example the bond dimension could be lowered to $2$ by re-ordering the basis; an example where the exponential bound is saturated {\it regardless} of the ordering of the basis is given in~\cite{graswald2021electronic}.)
Here we are interested in the best bond dimensions achievable by choosing the basis optimally, i.e.~performing an optimal fermionic mode transformation.

\begin{figure}[h!]
    \centering
    \includegraphics[width = 0.85\textwidth]{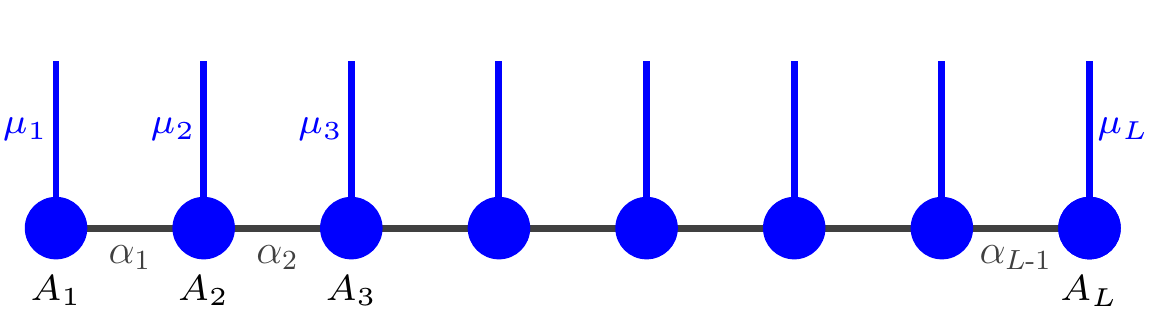}
    \caption{Graphical representation of a matrix product state in the finite-dimensional case; the virtual indices $\alpha_j$ are contracted over.}
    \label{fig:tt_representation}
\end{figure}


\section{Low-rank representation of two-electron wavefunctions and exactness of QC-DMRG with mode optimization}\label{sec:two_particle_finite}

We now show that for representing two-electron wavefunctions in the MPS format, bond dimension Three always suffices independently of $L$, provided the basis of the single-particle space is chosen optimally. 

In the following, the $N$-particle Hilbert space ${\mathcal V}_{N,L}=\bigwedge_{i=1}^N \H_L$ will be identified with the $N$-particle sector 
$$
   \{ \Psi\in\F_L \, : \, {\mathcal N}\Psi=N\Psi\}
$$
of Fock space, to which it is canonically isomorphic. Here ${\mathcal N}=\sum_{i=1}^L a^\dagger(\varphi_i)a(\varphi_i)$ is the number operator, with $a^\dagger(\varphi)$ and $a(\varphi)$ denoting the usual creation and annihilation operators associated with an orbital $\varphi\in \H_L$. Also, we will make use of the single-particle reduced density matrix $\gamma_\Psi \, : \, \H_L\to\H_L$ defined by 
$$
   \langle \Psi, \, a^\dagger(\varphi_i)a(\phi_j) \Psi \rangle = \langle \phi_j, \, \gamma_\Psi \phi_i \rangle \; \; \; \mbox{ for all }i,j.
$$

Our precise result on two-particle states is as follows.

\begin{theorem}[Upper bound on the bond dimensions]\label{thm:upper_bound}
For any two-particle state $\Psi\in {\mathcal V}_{2,L}$ with $L\geq 4$, there exists a basis $\{\varphi_1,...,\varphi_L\}$ of the  single-particle Hilbert space $\H_L$ for which $\Psi$ is an MPS with bond dimensions
\begin{equation}\label{eq:optimal_rank_vector}
    (r_1,...,r_{L-1}) =  (2,\underbrace{2,3,\ldots,2,3}_{L-4 \text{ times}},2,2).
\end{equation}
If $L\le 3$, we can simply achieve $(r_1,...,r_{L-1})=(1,...,1)$.
\end{theorem}

The somewhat counterintuitive looking bond dimension vector in \eqref{eq:optimal_rank_vector} is in fact optimal for generic two-particle wavefunctions. 

\begin{theorem}[Lower bound on the bond dimensions]\label{thm:optimality}

Suppose that $L\ge 4$ is even, $\Psi\in{\mathcal V}_{2,L}$, and $\gamma_\Psi$ has maximal rank (i.e., its rank equals $L$). Then the 
bond dimensions given in Theorem \ref{thm:upper_bound} are optimal, that is to say for any basis $\{\varphi_1,....,\varphi_L\}$ of the single-particle Hilbert space $\H_L$ and any MPS-representation with bond dimensions  $(r_1,\ldots,r_{L-1})$ we have 
\begin{itemize}
    \item $r_j \geq 2$ for every $j \in \{1,\ldots, L-1\}$
    \item At least one of two consecutive elements $(r_j,r_{j+1})$ for $j \in \{2,\ldots, L-2\}$ is at least  3.
\end{itemize}
Furthermore the bond dimension vector $(r_1,...,r_{L-1})$ with lowest $\ell^1$-norm $r_1+...+r_{L-1}$ is unique and given by \eqref{eq:optimal_rank_vector}.
\end{theorem}
As will become clear in the proof of Theorem \ref{thm:upper_bound}, the optimal representation is achieved for a basis consisting of natural orbitals, i.e. eigenstates of $\gamma_{\Psi}$. 

These results have an important implication for the QC-DMRG method for computing the electronic structure of molecules. 

\subsection{QC-DMRG method} This method approximates, for a given $N$-electron system, a given self-adjoint and particle-number-conserving Hamiltonian $H \, : \, \F_L\to\F_L$, and a given $L$-dimensional one-particle Hilbert space $\H_L$, the ground and excited energy levels and eigenstates of the system as follows: for a given basis $\{\varphi_1,...,\varphi_L\}$ of $\H_L$, 
\begin{equation}
\label{eq:qc_dmrg_enrgies}
   E_0^{\rm QC-DMRG}(\varphi_1,...,\varphi_L) = \min\limits_{\substack{\Psi\in\MPS(L,\{r_i\}_i,\{\varphi_i\}_i) \\ \Psi\neq 0, \, \N\Psi = N\Psi}} 
   \frac{\langle \Psi, H \Psi \rangle}{\langle \Psi,  \Psi \rangle}
\end{equation}
and 
\begin{equation}
\label{eq:qc_dmrg_enrgies_excited}
   E_j^{\rm QC-DMRG}(\varphi_1,...,\varphi_L) = \min\limits_{\substack{\Psi\in\MPS(L,\{r_i\}_i,\{\varphi_i\}_i) \\ \Psi\neq 0, \, \N\Psi = N\Psi, \\ \langle\Psi_k^{\rm QC-DMRG},\Psi\rangle=0\, \forall k=0,\ldots,j-1}} 
   \frac{\langle \Psi, H \Psi \rangle}{\langle \Psi,  \Psi \rangle} \;\;\; (j\ge 1),
\end{equation}
with the $\Psi_j^{\rm QC-DMRG}$ being corresponding optimizers. Our notation emphasizes that these quantities depend on the chosen single-particle basis. The exact (full configuration-interaction or FCI) eigenvalues $E_j$ and eigenstates $\Psi_j$ in the finite one-body basis are given, thanks to the Rayleigh-Ritz variational principle, by the analogous formulae with the MPS set $\MPS(L,\{r_i\}_i,\{\varphi_i\}_i)$ replaced by the full Fock space $\F_L$. 

\subsection{Mode transformations}
In recent simulations \cite{krumnow-legeza2016} it has been found to be beneficial to also optimize over the underlying one-body basis, i.e.~the `modes' $\varphi_1,...,\varphi_L$. Mathematically this corresponds to the following improved approximation to the eigenvalues and eigenstates:
\begin{align}\label{eq:mode_optimized}
E_0^{\rm QC-DMRG-MO}
=
\min\limits_{\substack{(\varphi_1,...,\varphi_L)\in \H_L\times\ldots\times \H_L \, : \\ \langle \phi_i,\phi_j\rangle = \delta_{ij} \, \forall i, \, j}}
E_0^{\rm QC-DMRG}(\varphi_1,...,\varphi_L)
\end{align}
and 
\begin{align}\label{eq:mode_optimized_excited}
 E_j^{\rm QC-DMRG-MO}
 =
 \min \limits_{\substack{(\varphi_1,...,\varphi_L)\in\H_L\times\ldots\times\H_L \, : \\ \langle \phi_i,\phi_j\rangle = \delta_{ij} \, \forall i, \, j}}
 E_j^{\rm QC-DMRG}(\varphi_1,...,\varphi_L) \;\;\; (j\ge 1),
\end{align}
where the superscript MO stands for mode-optimized. Corresponding optimizers in \eqref{eq:qc_dmrg_enrgies}, \eqref{eq:qc_dmrg_enrgies_excited} with optimal $\varphi_i$'s are denoted $\Psi_0^{\rm QC-DMRG-MO}$ respectively $\Psi_j^{\rm QC-DMRG-MO}$. Note that  such optimizers exist, since sets of normalized MPS states with given bond dimensions are closed \cite{holtz2012manifolds, barthel2021closedness} and bounded, and hence compact. 

Obviously, we have the inequalities
$$
    E_j \le E_j^{\rm QC-DMRG-MO}\le E_j^{\rm QC-DMRG}(\varphi_1,...,\varphi_N) \;\;\; \forall j\ge 0.
$$
Note also that for bond dimension $1$, i.e.~$r_i=1$ for all $i$, the QC-DMRG-MO ground state energy reduces precisely to the famous Hartree-Fock energy defined by
$$
   E_0^{HF} = \min\limits_{\substack{(\phi_1,\ldots, \phi_N) \in \H_L \times \ldots \times \H_L \\ \langle  \phi_{i}, \phi_j \rangle = \delta_{ij}}} 
   \frac{\langle \varphi_1 \wedge \ldots \wedge \varphi_N , H\varphi_1 \wedge \ldots \wedge \varphi_N\rangle}{\langle\varphi_1 \wedge \ldots \wedge \varphi_N, \,\varphi_1 \wedge \ldots \wedge \varphi_N\rangle},
$$
that is to say 
$$
  E_0^{\rm QC-DMRG-MO}\Big|_{r_1=...=r_{L-1}=1} = E_0^{HF}.
$$

The following interesting result is an immediate consequence of Theorem \ref{thm:upper_bound}.

\begin{theorem}[Low-rank exactness of QC-DMRG]\label{thm:exact}
For $N=2$ electrons, any particle-number-conserving self-adjoint Hamiltonian $H$, and any finite-dimensional single-particle Hilbert space ${\H_L}$, the QC-DMRG method with fermionic mode optimization is exact for bond dimension Three. That is to say,
\[
 E_j^{\rm QC-DMRG-MO}\Big|_{r_1=...=r_{L-1}=3} = E_j, \quad \forall j \geq 0,
\]
and any corresponding optimizers $\Psi_j^{\rm QC-DMRG-MO}$ are exact eigenstates.
\end{theorem}

\subsection{Necessity of mode optimization}
Mode optimization is essential for Theorems \ref{thm:upper_bound} and  \ref{thm:exact}, as the following example demonstrates.

\begin{example}

Consider an arbitrary fixed underlying basis $\{\phi_i\}_{i=1}^L$. As recalled in Lemma \ref{L:rank} below, the  minimal bond dimensions of a state $\Psi$ correspond to the ranks of the unfolding of its coefficient tensor $C$. 
These take the form 
\begin{align} \label{eq:unfolding_max_bonddim}
C^{\mu_{1}, \ldots,\mu_k}_{\mu_{k+1},\ldots, \mu_L}
=
    \left(
    \begin{array}{c|c|c} 
      && v_{1,k} \\ 
      \hline 
     & D_k & \\
      \hline
      v_{2,k} &&
    \end{array} 
    \right),
\end{align}
where $D_k \in \C^{k\times L-k}$ contains the coefficients corresponding to one $\mu_i =1$ in the $k$ upper indices and one $\mu_j =1 $ in the lower $L-k$ indices, analogously for $v_{1,k} \in \C^{1\times \binom{k}{2}}$ and $ v_{2,k} \in \C^{ \binom{L-k}{2} \times 1}$.

Thus a generic state $\Psi = \sum C_{\mu_1 \ldots \mu_L} \Phi_{\mu_1 \ldots \mu_L} \in \V_{2,L}$, resulting e.g.~from its coefficients being drawn independently from a continuous probability distribution -- like a standard Gaussian -- will have minimal bond dimensions 
$$
     r_k= 2+ \min\{k,L-k\}, 
$$ 
see \cite{feng2007rank}.
\end{example}
In this example, the overall bond dimension necessary, $\max_k r_k=2\! +\!\tfrac{L}{2}$, grows with the number of orbitals $L$. Note also that the state $\Psi$ above arises as the ground state of the parent Hamiltonian given by minus the orthogonal projector onto the state, that is, $H = - |\Psi\rangle \langle\Psi|$. Further, it follows from the results in \cite{graswald2021electronic} that there always exist states in ${\cal V}_{2,L}$ for which the overall bond dimension $2 \! + \! \tfrac{L}{2}$ cannot be reduced by re-ordering the basis.

\subsection{Upper bounds on the ranks \label{sec:upper_bound}}

In this subsection we prove Theorem \ref{thm:upper_bound}. We begin by recalling the following well known result \cite{coleman2000reduced}. 
\begin{lemma}[Two-particle wave-functions]
For any two-particle wavefunction $\Psi \in \V_{2,L}=\H_L\wedge \H_L$ there exists a basis $\{\phi_i\}_{i=1}^L$ of $\H_L$ and coefficients $\lambda_i$ ($i=1,...,k$), $k\le L/2$, such that
\begin{equation}\label{eq:two_particle_rep}
   \Psi = \sum_{\ell=1}^{k} {\lambda}_{\ell} ~ |  {\phi}_{2\ell-1} ,  {\phi}_{2\ell} \rangle,
\end{equation}
i.e.~each basis function appears only in one Slater determinant.
\end{lemma}
This can be proved by using the antisymmetry of $\Psi$ to write it in the form
\begin{equation} \label{eq:cij}
    \Psi = \sum_{1\leq i \neq  j \leq L}c_{ij} |\phi_i \phi_j \rangle
\end{equation}
with $c_{ij} = - c_{ji}$, and applying spectral theory to the coefficient matrix. We note that the orbitals appearing in \eqref{eq:two_particle_rep} are automatically  of $\Psi$, i.e.~\textit{natural orbitals} or norbs.

In the following we will always assume the basis to be chosen such that $\Psi$ is of the form \eqref{eq:two_particle_rep}.


The coefficient tensor in the occupation representation then takes the following form
\begin{equation} \label{eq:Ctens}
C_{\mu_1, \ldots, \mu_L} 
=
\sum
\limits_{\ell=1}^{k}
\lambda_\ell 
\delta_{11}^\ell
\prod
\limits_{\substack{i=1\\i\not = \ell}}^{k}
\delta_{00}^i,
\end{equation}
where we introduced the short-hand notation
\[
\delta_{00}^n := \delta_{0}(\mu_{2n-1}) \delta_{0}(\mu_{2n}),
\qquad
\delta_{11}^n := \delta_{1}(\mu_{2n-1}) \delta_{1}(\mu_{2n}).
\]
Due to this special structure
it makes sense to first seek a pair states decomposition, i.e.~an MPS factorization of $C_{\mu_1, \ldots, \mu_L}$ into tensors $B_\ell$ associated with pairs $(\mu_{2\ell-1},\mu_{2\ell})$ of occupation numbers, i.e.~$B_{\ell}\in \C^{r_{2\ell-2} \times 4 \times r_{2\ell}} $.

With respect to pair states, \eqref{eq:Ctens} looks like a non-translation-invariant version of the W-state $|\uparrow \downarrow \downarrow  ... \downarrow \rangle + |\downarrow \uparrow\downarrow ... \downarrow \rangle + ... + |\downarrow\downarrow\downarrow...\uparrow\rangle$ from spin physics, which is known to have bond dimension $2$. The following lemma gives a corresponding low-bond-dimension factorization in the non-translation-invariant case.

\begin{lemma}[Matrix lemma]
For any two sequences $(a_n)_{n\geq2}$ and $(b_n)_{n\geq2}$ of complex numbers, 
\begin{align}\label{eq:matrix_lemma}
    \begin{pmatrix}
    a_2 & b_2 \\
    &a_2
    \end{pmatrix}
       \begin{pmatrix}
    a_3 & b_3 \\
    &a_3
    \end{pmatrix}
    \cdots
       \begin{pmatrix}
    a_{n-1} & b_{n-1} \\
    &a_{n-1}
    \end{pmatrix}
    =
       \begin{pmatrix}
    \prod \limits_{i=2}^{n-1} a_i &  \sum \limits _{i=2}^{n-1}b_i\prod \limits_{j \neq i} a_j \\
    & \prod \limits_{i=2}^{n-1} a_i
    \end{pmatrix}.
\end{align}
\end{lemma}
\begin{proof}
We can write the left hand side of \eqref{eq:matrix_lemma} as 
\begin{align*}
    \left[
    a_2 \Id + b_2 \underbrace{\begin{pmatrix}
    0 & 1 \\  0 & 0
    \end{pmatrix}
    }_{=:S}
    \right]
       \left[
    a_3 \Id + b_3 \begin{pmatrix}
    0 & 1 \\  0 & 0
    \end{pmatrix}
    \right]
    \cdots    \left[
    a_{n-1} \Id + b_{n-1}  \begin{pmatrix}
    0 & 1 \\  0 & 0
    \end{pmatrix}
    \right].
\end{align*}
By the nilpotence of the matrix $S$, this becomes
\[
 \prod \limits_{i=2}^{n-1} a_i \Id + \sum \limits _{i=2}^{n-1}b_i\prod \limits_{j \neq i} a_j ~S,\]
 which is our assertion.
\end{proof}
Consequently, letting
\begin{eqnarray}
B_1[\mu_1, \mu_2] & := & \begin{pmatrix}
\delta_{00}^1  & ~\lambda_{1}\,\delta_{11}^1
\end{pmatrix}, \nonumber \\[2.5mm]
B_{\ell}[\mu_{2\ell-1}, \mu_{2\ell}] 
&:= &
\begin{pmatrix}
\delta_{00}^{2\ell}  &   \lambda_\ell ~\delta_{11}^{2\ell} \\
 & \delta_{00}^{2\ell} 
\end{pmatrix}  \;\; (1 < \ell < k), \label{eq:B} \\[1mm]
B_{k}[\mu_{2k-1}, \mu_{2k}]&:= &
\begin{pmatrix}
\lambda_k ~\delta_{11}^{2k} 
 \\
~~~\delta_{00}^{2k} 
\end{pmatrix} \nonumber
\end{eqnarray}
\noindent
we obtain the following MPS representation: 
\begin{align*}
C_{\mu_1,\ldots,\mu_L} &= B_1[\mu_1,\mu_2]  \ldots B_k[\mu_{2k-1},\mu_{2k}] \\
&=
    \begin{pmatrix}
\delta_{00}^1 & \lambda_{1} \delta_{11}^1
\end{pmatrix}
\begin{pmatrix}
\prod\limits_{\ell=2}^{k-1} \delta_{00}^\ell &  
\sum\limits_{\ell=2}^{k-1}\lambda_\ell \delta_{11}^\ell \prod\limits_{\substack{i=2\\i\not = \ell}}^{k-1}\delta_{00}^i
\\
& \prod\limits_{\ell=2}^{k-1} \delta_{00}^\ell
\end{pmatrix}
\begin{pmatrix}
\lambda_k \delta_{11}^k\\
\delta_{00}^k 
\end{pmatrix}\\
&=
\begin{pmatrix}
\prod\limits_{\ell=1}^{k-1} \delta_{00}^\ell \;  & \; 
\sum\limits_{\ell=1}^{k-1}\lambda_\ell \delta_{11}^\ell \prod\limits_{\substack{i=1\\i\not = \ell}}^{k-1}\delta_{00}^i
\end{pmatrix}
\begin{pmatrix}
\lambda_k \delta_{11}^k\\
\delta_{00}^k 
\end{pmatrix}
\\
&=
\sum\limits_{\ell=1}^{k}\lambda_\ell \delta_{11}^\ell \prod\limits_{\substack{i=1\\i\not = \ell}}^{k}\delta_{00}^i.
\end{align*}
\begin{figure}[ht!]
    \centering
    \includegraphics[width = 0.75\textwidth]{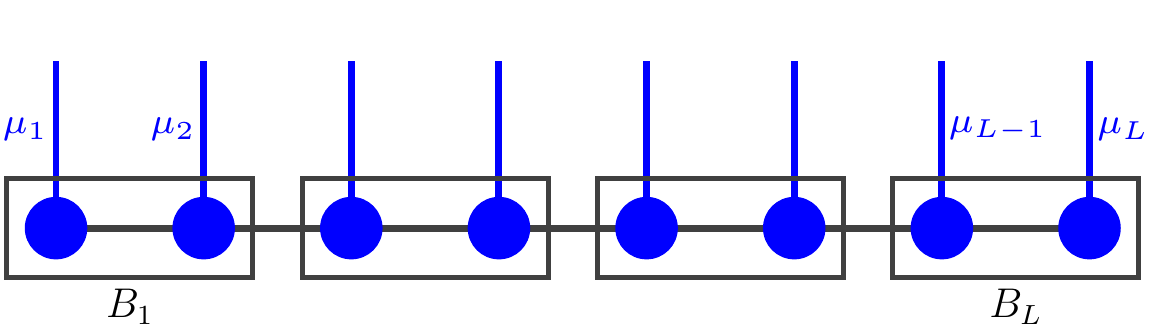}
    \caption{Graphical representation of the MPS decomposition associated with orbital pairs.}
    \label{fig:pair_states}
\end{figure}

The last step consists now in passing from the pair states to the original states, i.e.~decomposing the tensors $B_\ell[\mu_{2\ell-1},\mu_{2\ell}]$ into two tensors depending only on one of the $\mu_i$'s. 
This can be either guessed directly or obtained via reshaping and carrying out a singular value decomposition as in the derivation of the MPS representation of a general state (see e.g. \cite{schollwock2011density}). 

The result is 

\resizebox{1.15\linewidth}{!}{
  \begin{minipage}{\linewidth}
\begin{align*}
A_1[\mu_1] &:= \begin{pmatrix}
\delta_{0}(\mu_1) & \delta_{1}(\mu_1)
\end{pmatrix},
&
A_2[\mu_2] &:= \begin{pmatrix}
\delta_{0}(\mu_2) & \\
&
\lambda_{1} \delta_{1}(\mu_2)
\end{pmatrix},\\
A_{2\ell-1}[\mu_{2\ell-1}] 
&:=
\begin{pmatrix}
\delta_{0}(\mu_{2\ell-1}) & \lambda_\ell \delta_{ 1}(\mu_{2\ell-1}) & \\
& & \delta_{0}(\mu_{2\ell-1})
\end{pmatrix},
&
A_{2\ell}[\mu_{2\ell}] &:=\begin{pmatrix}
\delta_{0}(\mu_{2\ell}) & \\
& \delta_{1}(\mu_{2\ell}) \\
& \delta_{0}(\mu_{2\ell})
\end{pmatrix} \;\; (1<\ell<k), \\
A_{2k-1}[\mu_{2k-1}]&:=
\begin{pmatrix}
\lambda_{k}\delta_{1}(\mu_{2k-1}) &\\
&  \delta_{0}(\mu_{2k-1})
\end{pmatrix},
&
A_{2k}[\mu_{2k}]&:=
\begin{pmatrix}
\delta_{1}(\mu_{2k}) \\ \delta_{0}(\mu_{2k})
\end{pmatrix}.
\end{align*}
  \end{minipage}
}
\vspace*{1mm}

Therefore we have found an MPS representation for $\Psi$ with bond dimensions $r = (2,2,3,2,\ldots,2,3,2,2)$.
Note that in the case $L=4$, i.e.~$k=2$, this reduces to $r=(2,2,2)$.
Finally, in the case of just one Slater-determinant, i.e.~$L=2$, we can use 
$a_1[\mu_1]:= \lambda_1 \delta_{1}(\mu_1)$, $a_2[\mu_{2}]:= \delta_{1}(\mu_2)$. 
This completes the proof of Theorem \ref{thm:upper_bound}.

\subsection{Lower bounds on the ranks}
In the previous subsection we saw that we can choose a basis such that all states can be represented with bond dimensions $r= (2,\underbrace{2,3,\ldots,2,3}_{L-4 \text{ times}},2,2)$. 
But as the product of the matrices of two orbitals can be written as a $2\times2$-matrix (see the $B_\ell$ above)
one might wonder if the maximal bond dimension can be brought down to 2.
This turns out not to be the case and the above size vector is optimal as stated in Theorem \ref{thm:optimality}.

Our starting point to prove Theorem 2 is the following well known fact.

\begin{lemma}[TT-rank equals separation rank \cite{holtz2012manifolds}, \cite{hackbusch2012tensor}
]\label{L:rank}
Let $C \in \C^{n_1 \times\ldots \times n_d}$ be an arbitrary tensor (representing the coefficients of a quantum state with respect to a fixed basis). 
For each bond between the $i$th and the $i+1$st matrix, there exists a minimial $r_i$ such that $C$ admits a TT-decomposition with $A_i$ of size $n\times r_i$ and $A_{i+1}$ of size $r_i\times m$, and this $r_i$ is given by the Schmidt rank of the unfolding $C^{\mu_1..\mu_{i}}_{\mu_{i+1}...\mu_L}$. Also, there exists a TT-decomposition with all $r_i$ being simultaneously minimal.
\end{lemma}

\begin{proof}[Proof of Theorem \ref{thm:optimality}]
We start with the case $L=4$ to convey the proof idea.
Due to Lemma \ref{L:rank} it is enough to consider the unfoldings $C^{\mu_1}_{\mu_1\mu_2\mu_3}$ and $C^{\mu_1\mu_2}_{\mu_3\mu_4}$ of our tensor $C_{\mu_1 \mu_2 \mu_3 \mu_4}$.
Note that here the coefficients $c_{ij}$ in \eqref{eq:cij} are with respect to some arbitrary underlying basis $\{\phi_{i}\}_{i=1}^L$.
We begin with 
\begin{align*}
    C^{\mu_1}_{\mu_2,\mu_3,\mu_4}
    =
    \kbordermatrix{
    & 110 & 101 & 101 & 100 & 010 & 001 \\
    0 & c_{23} &c_{24} &c_{34} & & &\\
    1 & & & & c_{12} & c_{13} & c_{14}    
    }
\end{align*}
The second row cannot vanish since then the state $\phi_1$ would not appear at all, meaning $\gamma_\Psi$ has rank $<4$, a contradiction.\\
And if the first row vanishes we can define 
$
\tilde \phi_{2} := \sum\limits_{k=2}^4 c_{1,k}\phi_k
$ and thus get 
\[
\Psi = | \phi_1,  \tilde \phi_2 \rangle,
\]
so again $\mbox{rank}\,\gamma_\Psi<4$.
So the matrix $C^{\mu_1}_{\mu_2,\mu_3,\mu_4}$ must have rank $2$.

Next let us consider
\begin{align*}
    C^{\mu_1, \mu_2}_{\mu_3,\mu_4}
    =
    \kbordermatrix{
     & 00 & 10 & 01 & 11 \\
    00 &  & &&c_{34} \\
    10 & & c_{13} & c_{14}  &\\
    01 & &  c_{23} & c_{24} &\\
    11 &c_{12} & & & 
    }.
\end{align*}
We want to show that  $\rk  C^{\mu_1, \mu_2}_{\mu_3,\mu_4} \geq 2$.

If the submatrix in the middle vanishes, then both $c_{12} \not = 0$ and $c_{34} \not = 0$, otherwise $L \leq 2$. But then $\rk  C^{\mu_1, \mu_2}_{\mu_3,\mu_4} = 2$.
So we can assume that the submatrix in the middle does not vanish. If it has rank $2$ we are already done. Thus assume that the rank equals $1$ and $c_{12} = c_{34} =0$. Then we know that there is a $\lambda \in \C$ such that 
\[
\begin{pmatrix}
c_{14} \\
c_{24}
\end{pmatrix}
=
\lambda
\begin{pmatrix}
c_{13} \\
c_{23}
\end{pmatrix}
\]
but then we can write $\Psi$ as
\[
\Psi 
=
 c_{13} |\phi_{1}, \phi_{3} + \lambda \phi_4 \rangle +  c_{23} |\phi_{2}, \phi_{3} + c \phi_4 \rangle, 
\]
so $L<4$.
The unfolding $   C^{\mu_1, \mu_2,\mu_3}_{\mu_4}$ is dealt with in the same way as $   C^{\mu_1
}_{\mu_2,\mu_3,\mu_4}$.
Therefore if $\rk\gamma_\Psi = L = 4$, then the lowest possible rank vector is $r= (2,2,2).$

Let us now turn to the general case $L \geq 6$. We start by noting that the first unfolding $C^{\mu_1}_{\mu_2,\ldots,\mu_L}$ and the last unfolding $C^{\mu_1,\ldots,\mu_{L-1}}_{\mu_L}$ both always have rank $2$. The argument is exactly the same as in the $L=4$ case. Thus we already know $r_1= r_{L-1}=2$.
Consider now the unfoldings
\begin{align*}
      M_{n}&:= C^{\mu_1,\ldots,\mu_{n}}_{\mu_{n+1},\ldots\mu_L}
  =
  \kbordermatrix{
 & 0\ldots0 & 10\ldots0 &\cdots & & \cdots & 0\ldots01 & 110\ldots0 & \cdots & & \cdots &0\ldots 011\\
 0\ldots 0& & & & & & & c_{n+1,n+2} & \ldots & &\ldots  &c_{L-1,L} \\
 10\ldots0 & & c_{1,n+1} & \ldots & & \ldots & c_{1L}\\
 \vdots & & \vdots & & & & \vdots \\
 \\
\vdots & & \vdots & & & & \vdots \\
 0\ldots01 & &c_{n,n+1} & \ldots & & \ldots & c_{nL} \\
 110\ldots0 & c_{12}\\
 \vdots & \vdots\\
 \\
 \vdots & \vdots \\
 0\ldots011 & c_{n-1,n}
 }.
\end{align*}

We start by proving that these matrices $\big(M_{n}\big)_{n=2,\ldots,L-2}$ always have rank $\geq 2$.

Assume the first row vanishes.
If the submatrix corresponding to one $\mu_i$ being $1$ in the upper indices and one $\mu_j$ being $1$ in the lower indices -- denoted by $M^{11}_n$ -- has rank $\geq 2$, there is nothing to show. So assume that $\rk M^{11}_n \leq 1$. 
Then we can recombine the states $\phi_{n+1},\ldots,\phi_{L}$ to see that $\rk\gamma_\Psi < L$, as follows.
Since $\rk M^{11}_n \leq 1$, all columns are multiples of a single column, i.e.~w.l.o.g.
 \begin{align*}
     \exists \alpha_{j}:
\begin{pmatrix}
c_{1j}\\
\vdots\\
c_{nj}
\end{pmatrix}
=
\alpha_j
\begin{pmatrix}
c_{1,n+1}\\
\vdots\\
c_{n,n+1}
\end{pmatrix}
\quad \forall j\in \{n+1,\ldots,L  \}.
 \end{align*}
Then we can write $\Psi$ as 
\begin{align*}
    \Psi 
    &=
    \sum_{1\leq r < s \leq n} c_{rs} |\phi_r,\phi_s\rangle
    +
    \sum_{\substack{1\leq r \leq n \\ n+1 \leq s \leq L}} c_{rs} |\phi_r,\phi_s\rangle \\
    &=
     \sum_{1\leq r < s \leq n} c_{rs} |\phi_r,\phi_s\rangle
    +
    \sum_{1\leq r \leq n }  c_{r,n+1} \big|\phi_r,\underbrace{\sum_{n+1 \leq s \leq L} \alpha_{s}\phi_s}_{\tilde \phi_{n+1}} \big\rangle.
\end{align*}
So $\Psi $ can be represented with only at most $n+1$ basis functions, i.e.~$\rk \gamma_\Psi \leq n+1 < L$, a contradiction.

In the same way, assuming that the first column vanishes and that $\rk M_n^{11} \leq 1$, we obtain w.l.o.g.
 \begin{align*}
     \exists \beta_{j}:
\begin{pmatrix}
c_{j,n+1}, &
\ldots, & 
c_{jL}
\end{pmatrix}
=
\beta_j
\begin{pmatrix}
c_{1,n+1},&
\ldots, &
c_{1L}
\end{pmatrix}
\quad \forall j\in \{1,\ldots,n \}.
 \end{align*}
Then we can write $\Psi$ as 
\begin{align*}
    \Psi 
    &=
    \sum_{n+1\leq r < s \leq L} c_{rs} |\phi_r,\phi_s\rangle
    +
    \sum_{\substack{1\leq r \leq n \\ n+1 \leq s \leq L}} c_{rs} |\phi_r,\phi_s\rangle \\
    &=
      \sum_{n+1\leq r < s \leq L} c_{rs} |\phi_r,\phi_s\rangle
    +
     \sum_{n+1 \leq s \leq L} c_{1s} \big|\underbrace{\sum_{1\leq r\leq n }\beta_r \phi_r}_{\tilde \phi_n},~\phi_s\big\rangle.
\end{align*}
Consequently $\Psi $ can be represented with at most $L+1-n$ basis functions, i.e.~$\rk \gamma_\Psi \leq L+1-n < L$, so we again obtain a contradiction.
Hence we have proven that if the first row or the first column of $M_n$ vanish, the submatrix $M_n^{11}$ has rank $\geq 2$.
Since we have dealt with $M_1$ and $M_{L-1}$ separately, we thus have shown
\[
\rk M_n \geq 2 \quad \forall n \in \{1,\ldots, L-1\}.
\]

Our next steps now consists in considering two unfoldings at the same time and prove that at least one of them has rank $\geq 3$. 

Therefore, consider for every $\ell \in \{1,\ldots, k-2 \}$ the following matrices $M_{2 \ell}$ and $M_{2\ell +1}$, where $M_{i}:= C^{\mu_1,\ldots,\mu_{i}}_{\mu_{i+1},\ldots\mu_L}$ and $L=2k$. 
\begin{align*}
  M_{2\ell}:
  =
 \kbordermatrix{
 & 0\ldots0 & 10\ldots0 &\cdots & & \cdots & 0\ldots01 & 110\ldots0 & \cdots & & \cdots &0\ldots 011\\
 0\ldots 0& & & & & & & c_{2\ell+1,2\ell+2} & \ldots & &\ldots  &c_{L-1,L} \\
 10\ldots0 & & c_{1,2\ell+1} & \ldots & & \ldots & c_{1L}\\
 \vdots & & \vdots & & & & \vdots \\
 \\
\vdots & & \vdots & & & & \vdots \\
 0\ldots01 & &c_{2\ell,2\ell+1} & \ldots & & \ldots & c_{2\ell,L} \\
 110\ldots0 & c_{12}\\
 \vdots & \vdots\\
 \\
 \vdots & \vdots \\
 0\ldots011 & c_{2\ell-1,2\ell}
 }
 \end{align*}
 \begin{align*}
  M_{2\ell+1}:
  =
  \kbordermatrix{
 & 0\ldots0 & 10\ldots0 &\cdots & & \cdots & 0\ldots01 & 110\ldots0 & \cdots & & \cdots &0\ldots 011\\
 0\ldots 0& & & & & & & c_{2\ell+2,2l+3} & \ldots & &\ldots  &c_{L-1,L} \\
 10\ldots0 & & c_{1,2\ell+2} & \ldots & & \ldots & c_{1L}\\
 \vdots & & \vdots & & & & \vdots \\
 \\
\vdots & & \vdots & & & & \vdots \\
 0\ldots01 & &c_{2\ell+1,2\ell+2} & \ldots & & \ldots & c_{2\ell+1,L} \\
 110\ldots0 & c_{12}\\
 \vdots & \vdots\\
 \\
 \vdots & \vdots \\
 0\ldots011 & c_{2\ell,2\ell+1}
 }
\end{align*}
Note that with this range for $\ell$ we do not reach the unfolding $M_{L-2}:=C^{\mu_{1},\ldots,\mu_{L-2}}_{\mu_{L-1},\mu_{L}}$. This is to be expected since we always have $L-1$ unfoldings and the first and the last one have to be dealt with separately, so from the remaining  $L-3$ unfoldings -- which is an odd number -- one matrix will be left out. 

Assume now that both $M_{2\ell}$ and $M_{2\ell+1}$ only have rank $2$.
Above we showed that  if the first row or the first column of $M_n$ vanish the submatrix $M_n^{11}$ has rank $\geq 2$. Thus only two cases could happen for each $M_n$: either both the first row and the second row vanish and $\rk M_n^{11}=2$ (\textit{case 1}) or both the first row and the first column do not vanish and $\rk M_n^{11} = 0$ (\textit{case 2}).

Note that if for $M_{2\ell}$ \textit{case 1} occurs then clearly also the first row of $M_{2\ell+1}$ vanishes so either $\rk M_{2\ell+1} \geq 3$ or also for $M_{2\ell+1}$ \textit{case 1} happens.
Similarly if $M_{2\ell+1}$ falls into \textit{case 1}, then the first column  of $M_{2\ell}$ vanishes 
 so either $\rk M_{2\ell} \geq 3$ or again both matrices satisfy \textit{case 1}.
 
Therefore we only need to check the following two overall situations.

First, assume that for both $M_{2\ell}$ and $M_{2\ell+1}$ \textit{case 1} occurs. Then $c_{j,2\ell+1} = c_{2\ell+1,j} = 0$ for all $j$, i.e.~the state $\phi_{2\ell+1}$ does not appear in $\Psi$, so $\rk \gamma_\Psi < L$, a contradiction. To see this note that all $c_{j,2\ell+1}$ are contained in the first column of $M_{2\ell+1}$ and all $c_{2\ell+1,j}$ are contained in the first row of $M_{2\ell}$.

Second, assume that for both $M_{2\ell}$ and $M_{2\ell+1}$ \textit{case 2} occurs. Then as above we obtain $c_{j,2\ell+1} = c_{2\ell+1,j} = 0$ for all $j$, i.e.~$\rk \gamma_\Psi < L$. This time the vanishing of the coefficients stems from the fact that  all $c_{j,2\ell+1}$ are contained in the first column of $M^{11}_{2\ell}$ and all $c_{2\ell+1,j}$ are contained in the last row of $M^{11}_{2\ell + 1}$.

In conclusion we have shown that for $\big(M_{n} \big)_{n = 2}^{L-2}$ one of two consecutive  matrices must always have $\rk \geq 3$.

Since $(r_2, \ldots, r_{L-2})$ has an odd number of entries the lowest possible ranks are the ones starting with $2$ and not with $3$, i.e.
\[
(r_2, \ldots, r_{L-2}) 
=
(2,3,2,3, \ldots, 2,3,2),
\]
which yields the lowest possible rank vector
\[
(r_1,...,r_{L-1}) =  (2,\underbrace{2,3,\ldots,2,3}_{L-4 \text{ times}},2,2).
\]
The proof of Theorem \ref{thm:upper_bound}
is complete.
\end{proof}

\section{Matrix product states -- Infinite dimensions}

We now deal with infinite-dimensional single-particle Hilbert spaces $\H$. As we will see, this calls for half-infinite matrix product states which we will introduce in a rigorous manner below. 
Graphically this corresponds to a half-infinite chain, see Figure \ref{fig:tt_rep_infinite}.

So let $\H$ be an infinite-dimensional separable Hilbert space spanned by orthonormal orbitals 
 $\{\phi_i\}_{i=1}^{\infty}$, let ${\mathcal V}_N$ be the $N$-fold antisymmetric product $\bigwedge_{i=1}^N \H$, and let $\F$ be the ensuing Fock space, 
 $$
    \F := \bigoplus_{N=0}^\infty {\mathcal V}_N.
 $$
 
 Analogously to \eqref{eq:MPS}, we define a matrix product state (MPS) or tensor train (TT) with respect to the basis $\{\varphi_i\}_{i=1}^\infty$ 
with size parameters ('bond dimensions') $\{r_i\}_{i=1}^\infty$ to be a state of the form
\begin{equation} \label{eq:MPS_infinite}
   \Psi = \lim_{L\to\infty} \sum_{\mu_1,\ldots \mu_L =0}^1  \! A_1[\mu_1] A_2[\mu_2] ... A_L[\mu_L] 
   \begin{pmatrix} 0  \\ \vdots \\ 0 \\ 1 \end{pmatrix}
   ~\Phi_{\mu_1 ... \mu_L} \in \F,
\end{equation}
where the $A_i[\mu_i]$ ($i=1,2,...)$ are $r_{i-1} \times r_i$ matrices, $r_0=1$, the column vector above has length $r_L$ (so as to make the coefficient of $\Phi_{\mu_1...\mu_L}$ scalar), and the $A_i$ are such that the above limit exists as a strong limit in the Fock space $\F$. The key point about the representation \eqref{eq:MPS_infinite} is that the $A_i$ are {\it fixed} matrices which only depend on the {\it exact} infinite-dimensional quantum state $\Psi$ and encode its true entanglement structure, whereas first truncating the one-body Hilbert space to dimension $L$ and then MPS-factorizing the ensuing approximation to $\Psi$ would lead to $L$-dependent $A_i$'s.

The vector $(0,...,0,1)$ appearing in \eqref{eq:MPS_infinite} may look arbitrary at first, but as we show in a companion paper \cite{friesecke_graswald_infinite_mps} every normalized state $\Psi$ in the Fock space $\F$ can be represented in the form \eqref{eq:MPS_infinite} with left-normalized $A_i$ (i.e. $\sum_{\mu_i}A_i(\mu_i)^\dagger A_i(\mu_i)=I$) if the $r_i$ are allowed to grow exponentially (i.e. $r_i=2^i$).

The set of tensor trains (TT) or matrix product states (MPS) with respect to the basis $\{\varphi_i\}_{i=1}^\infty$ with bond dimensions $\{r_i\}_{i=1}^\infty$ is denoted by 
\begin{equation}\label{eq:mps_representation_infinite}
    \mathrm{MPS}\bigl(\infty,\{r_i\}_{i},\{\varphi_i\}_{i}\bigr) \; \; \subseteq \;\; \F.
\end{equation}

\begin{figure}[h!]
    \centering
    \hspace*{10mm} \includegraphics[width = 0.85 \textwidth]{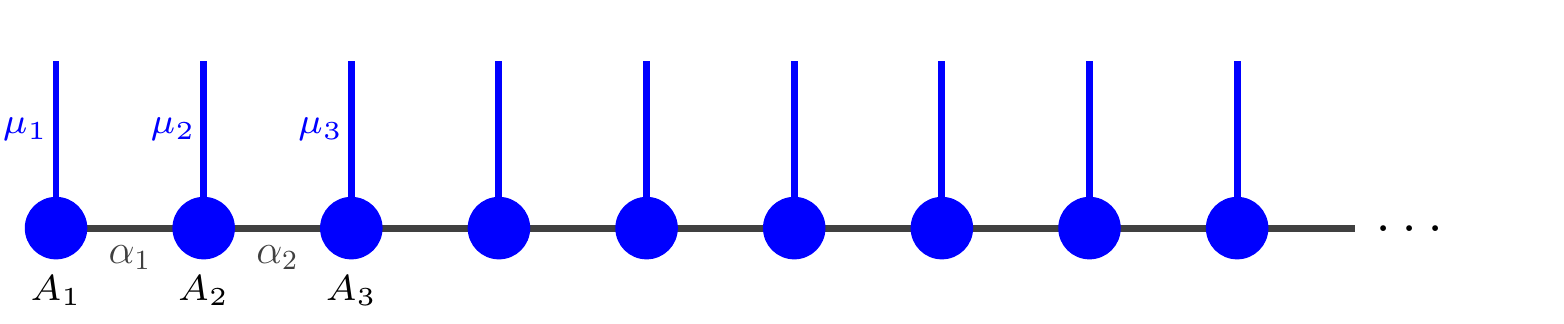}
    \caption{Graphical representation of a matrix product state in the infinite-dimensional case.}
    \label{fig:tt_rep_infinite}
\end{figure}




\section{Two-particle systems -- Infinite dimensions}

We now extend our results from section \ref{sec:two_particle_finite} to infinite dimensions.
%
%
\begin{theorem}\label{thm:onto} 
Let $\H$ be an infinite-dimensional separable Hilbert space. For any two-particle state $\Psi\in \H\wedge\H$, there exists an orthonormal basis $\{\varphi_i\}_{i=1}^\infty$ of the  single-particle Hilbert space $\H$ for which $\Psi$ is an MPS with bond dimensions $r_i=2$ for $i$ even, $r_i=3$ for $i$ odd and $>1$, and $r_1=2$. In particular, there is an MPS representation with maximal bond dimension $3$. 

Moreover the value $3$ is minimal, that is, not all two-particle states can be represented by an MPS with bond dimension $2$.
\end{theorem}


%

\begin{proof}
Let $\Psi$ be in $\H\wedge \H$, and let $\{\phi_i\}_{i=1}^\infty$ be an ONB of $\H$ consisting of eigenstates of $\gamma_\Psi$ (i.e., of natural orbitals),  ordered by size of the eigenvalue of $\gamma_\Psi$. 
After a unitary transformation in each eigenspace, $\Psi$ has the normal form 
\begin{align*}
    \Psi = \sum_{\ell =1}^\infty \lambda_\ell |\phi_{2\ell -1} \phi_{2 \ell} \rangle,
\end{align*}
with $\sum_{\ell=1}^\infty |\lambda_\ell|^2 = ||\Psi||^2 < \infty$. 
For $L$ even, define 
\begin{align*}
     \Psi_L = \sum_{\ell =1}^L \lambda_\ell |\phi_{2\ell -1} \phi_{2 \ell} \rangle. 
\end{align*}
 Applying now our analysis from Section \ref{sec:two_particle_finite} gives that $\Psi_L$ has a representation of the form \eqref{eq:mps_representation_finite} with {$L$-dependent} tensors $A_1^{(L)}[\mu_1], \ldots, A_L^{(L)}[\mu_L]$.
 
 By inspection these $A^{(L)}_i$ only depend on the coefficients {$\lambda_\ell$} up to $\lceil  \tfrac{i}{2}\rceil $, so they are independent of $L$ for $L \gg i$; thus denote these by $A_i$. As in Section \ref{sec:two_particle_finite}, let $B_\ell[\mu_{2\ell-1},\mu_{2\ell}]=A_{2k-1}[\mu_{2\ell-1}]A_{2\ell}[\mu_{2\ell}]$, whence $B_{\ell}$ is given by eq.~\eqref{eq:B}. Now let us compute the expression inside the limit in eq.~\eqref{eq:MPS_infinite}. When $L$ is even, that is, $L=2k$ for some integer $k$, we have
 $$
     B_k[\mu_{2k-1},\mu_{2k}] \begin{pmatrix} 0 \\ 1 \end{pmatrix} = \begin{pmatrix} 
     \lambda_k \delta_{11}^{2k}(\mu_{2k-1},\mu_{2k}) \\
     \delta_{00}^{2k}(\mu_{2k-1},\mu_{2k}) \end{pmatrix}
 $$
 and therefore 
 $$
     \Bigl(\prod_{\ell=1}^k B_{\ell}[\mu_{2\ell-1},\mu_{2\ell}]\Bigr) \begin{pmatrix} 0  \\ 1 \end{pmatrix} = \sum_{\ell=1}^k \lambda_\ell \delta_{11}^{\ell}\prod_{\substack{i=1, \\i\neq \ell}}^k \delta_{00}^i. 
 $$
 It follows that 
 \begin{equation} \label{eq:truncsum}
  \sum_{\mu_1,\ldots \mu_L =0}^1  \! A_1[\mu_1] A_2[\mu_2] ... A_L[\mu_L] 
   \begin{pmatrix} 0 \\ 1 \end{pmatrix}
   ~\Phi_{\mu_1 ... \mu_L} = \Psi_{2k}.
 \end{equation}
 For $L$ odd, that is, $L=2k+1$ for integer $k$, one finds analogously that
 $$
      \Bigl(\prod_{i=1}^{2k+1} A_{i}\Bigr)  \begin{pmatrix} 0  \\ 0 \\ 1 \end{pmatrix}
      = \Bigl(\prod_{i=1}^k B_i\Bigr) A_{2k+1}  \begin{pmatrix} 0 \\ 0 \\ 1 \end{pmatrix} = \sum_{\ell=1}^k \lambda_\ell \delta_{11}^{\ell} \! \prod_{\substack{i=1, \\i\neq \ell}}^k \! \delta_{00}^i \; \delta_0(\mu_{2k+1})
 $$
 and therefore the left hand side in eq.~\eqref{eq:truncsum} is again given by $\Psi_{2k}$. 
 Since $\Psi_L$ converges by construction to $\Psi$, we obtain
 \begin{align*}
     \Psi = \limit{L}{\infty} \Psi_L 
     =  \lim_{L\to\infty} \sum_{\mu_1,\ldots \mu_L =0}^1  \! A_1[\mu_1] A_2[\mu_2] ... A_L[\mu_L] \begin{pmatrix} 0 \\ \vdots \\ 0 \\ 1 \end{pmatrix} ~\Phi_{\mu_1 ... \mu_L}.
 \end{align*}
 Thus $\Psi$ has an MPS {representation} with the asserted bond dimensions with respect to our natural orbital basis. 
 
 The fact that $\Psi$ does not in general belong to the set of MPS with bond dimension $2$ regardless of the choice of basis follows directly from Theorem \ref{thm:optimality}.
\end{proof}

As a corollary, the exactness of the QC-DMRG method combined with fermionic mode transformations for two-electron systems (Theorem \ref{thm:exact}) generalizes in a straightforward manner to the full infinite-dimensional single-particle Hilbert space $\H=L^2(\R^3)\otimes\C^2$ for electrons. The resulting versions of \eqref{eq:qc_dmrg_enrgies}--\eqref{eq:mode_optimized_excited} constitute an exact reformulation of (time-independent) two-electron quantum mechanics.

\section{Conclusions and Outlook}
We have shown that the QC-DMRG method combined with fermionic mode optimization is exact for two-electron systems with the (extremely low) bond dimension $M=3$. This can be viewed as a theoretical contribution towards explaining the remarkable success of the QC-DMRG method in practical computations, and as a theoretical argument in favour of including mode optimization. The numerical favourability of the latter was emphasized in \cite{krumnow-legeza2016}. 

An interesting theoretical question beyond the scope of the present paper is whether any analoga of our findings hold for larger particle numbers provided the Hamiltonian is of two-body form. In the two-electron case investigated here, this form was satisfied automatically; in electronic structure it continues to be satisfied for arbitrary particle numbers.

\section*{Acknowledgements}

Support by the Deutsche Forschungsgemeinschaft (DFG, German Research Foundation) -- Project number 188264188/GRK1754 within the International Research Training Group IGDK 1754 
%
is gratefully acknowledged.


\end{document}